\documentclass{article}

\usepackage{graphicx}
\usepackage{amssymb}
\usepackage{xcolor}

\usepackage{amsthm}

\theoremstyle{plain}

\newtheorem{proposition}{Proposition}

\theoremstyle{definition}
\newtheorem{definition}{Definition}
\newtheorem{algorithm}{Algorithm}
\newtheorem{circuit}{Quantum circuit}

\begin{document}

\title{How much classical information is carried by a quantum state?
       An approach inspired by Kolmogorov complexity. }
\author{Doriano Brogioli} \date{Universit\"at Bremen,
  Energiespeicher-- und Energiewandlersysteme, Bibliothekstra{\ss}e 1,
  28359 Bremen, Germany.}

\maketitle

\begin{abstract}
  In quantum mechanics, a state is an element of a Hilbert space whose
  dimension exponentially grows with the increase of the number of
  particles (or qubits, in quantum computing). The vague question ``is
  this huge Hilbert space really there?'' has been rigorously
  formalized inside the computational complexity theory; the research
  suggests a positive answer to the question. Along this line, I give
  a definition of the (classical) information content of a quantum
  state, taking inspiration from the Kolmogorov complexity. I show
  that, for some well-known quantum circuits (having a number of gates
  polynomial in the number of qubits), the information content of the
  output state, evaluated according to my definition, is polynomial in
  the number of qubits. On the other hand, applying known results, it
  is possible to devise quantum circuits that generate much more
  complex states, having an exponentially-growing information content.
  A huge amount of classical information can be really present inside a
  quantum state, however, I show that this property is not necessarily
  exploited by quantum computers, not even by quantum algorithms
  showing an exponential speed-up with respect to classical
  computation.
\end{abstract}

\section{Introduction}

One of the surprising aspects of quantum mechanics is the enormous
size of the data necessary to represent a quantum state, in comparison
with the classical view of the same system. Although it underlyes
every physical phenomenon, most of this ashtonishing complexity is
lost when we observe the system at macroscopic level. As an example,
performing an \emph{ab initio} calculation of a semiconductor
electronic structure, by a brute-force numerical solution of the
Schr{\"o}dinger equations, is practically unfeasible; complex
approximation methods are needed to simulate the system in a
reasonable time. Notwithstanding the complexity of the calculation,
the result is just expressed by a few calculated parameters, such as
the band-gap: a few digits, which hardly show the ashtonishing
complexity of the underlying system.

The discussion above hinges around the word ``calculation'': actually,
a deep insight into the complexity of quantum mechanics is obtained by
experiments on quantum computers. They can be seen as laboratory
instruments for testing the most unusual quantum phenomena, rather
than as machines for solving difficult mathematical problems. In the
field of quantum computing, the concepts presented in the previous
paragraph are well formalized.  The information contained in a quantum
state is actually huge: the quantum equivalent of $n$ bits, i.e. $n$
qubits (the bits of quantum computers), is described by $2^n$ complex
numbers (actually, rational numbers suffice~\cite{aharonov2003}). This
notwithstanding, a measurement of the quantum state just gives us $n$
classical bits. What is surprising is rather that quantum computers
are exponentially faster (in $n$) than classical ones.  More formally,
this is expressed in computational complexity
theory~\cite{arora_barak} by the fact that the class {\bf BQP}
(roughly, problems that can be solved in polynomial time by a quantum
computer with good probability) is believed to strictly include the
class {\bf P} (roughly, the problems that can be solved in polynomial
time by classical computers).

The exploration of an exponentially large number of classical
configuration is often described as the reason of the exponential
speed-up in the calculation time; generating a large superposition of
states, with a large entanglement, and make them interfere are
necessary ingredients~\cite{vidal2003} (although not
sufficient~\cite{jozsa2003}) of quantum circuits showing superior
performances with respect to classical computers. This can be seen as
a quite indirect proof of the existence of an exponentially large
amount of information in a quantum state.

The quite vague question: ``are all these components of the quantum
state really there?'' arised in various contexts, in
particular connected with the interpretations of quantum
mechanics. More practically and rigorously, this question can been
formalized in computational complexity theory. S. Aaronson mentioned
this question as one of the ``ten grand challenges for quantum
computing theory''~\cite{aaronson_ten_challenges} and suggested to
address it by studying the substitution of classical certificates and
advices instead of their quantum equivalents in two quantum complexity
classes, {\bf QMA} and {\bf BQP/qpoly}~\cite{aaronson2007}. An even
more practical reason for asking the question is to devise
memory-compressed (zipped) representations of quantum states, aimed at
a more efficient simulation of quantum systems~\cite{hillmich2020}.

Actually, there are clues suggesting that the information contained in
a quantum state is not so huge. One of such clues is that {\bf BQP} is
a strict subset of {\bf PSPACE}: roughly, the quantum computers can be
simulated using a polynomial amount of memory, although a very long
calculation time could be needed. Unfortunately, this fact does not
help us to identify a compact structure representing the quantum
state. Another clues is that the quantum computers are less powerful
than ``non-deterministic Turing machines'', i.e. they do not
simultaneously explore every possible classical state (as sometimes
erroneously said in science popularization): this is formally
expressed by the belief that {\bf BQP} does not to contain {\bf
  NP-complete} problems, or by the fact that the Grover's search
algorithm only gives a polynomial (instead of an exponential)
speed-up.

\begin{figure}
  \includegraphics{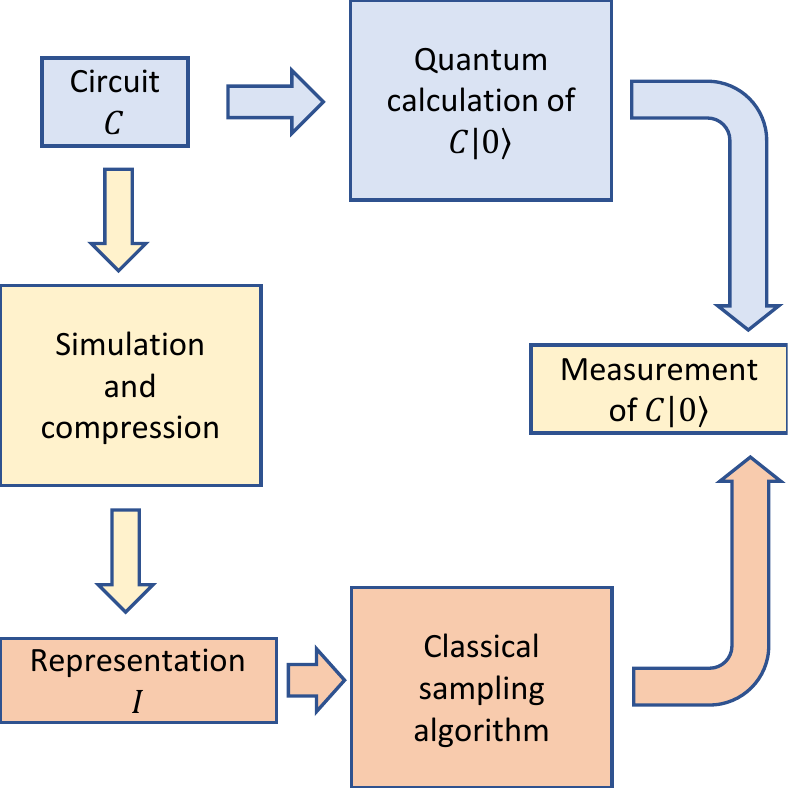}
  \caption{Scheme of the approach to the evaluation of the information
    content of a quantum state generated by a circuit. In the upper
    part, the circuit $C$ is implemented in a quantum computer,
    applied to the initial state $\left|0\right>$, and the output
    state is measured in the computational basis. In the lower part, a
    compressed (zipped) representation $I$ of $\left|0\right>$ is fed
    as input to an algorithm, which generates samples with the same
    probability distribution of the quantum computer.}
  \label{fig:scheme}
\end{figure}

In order to systematically connect the results available in literature
to the above-described question, here I propose an approach inspired
by the Kolmogorov complexity to evaluate the classical information
content of a quantum state generated by a circuit. In short, the idea
behind the Kolmogorov complexity is that the amount of information
carried by a string equals the length of the shortest algorithm able
to generate it. The approach must however be adapted, as schematically
shown in Fig.~\ref{fig:scheme}. Indeed, in order to extend this idea
to quantum states, two problems have to be solved.

The first problem is that the quantum state is not a string but rather
gives a probability distribution. There are two possible meanings of
``simulating a quantum system'': either generate the numbers
representing the probabilities of measuring a given output, or to
sample from that distribution, i.e. generate random numbers with the
same probability of the quantum state measurement.  These two
operations are called strong and weak simulation, respectively.  The
two operations do not have the same complexity, e.g. the weak is
faster than the strong simulation~\cite{bravy2016}. According to the
discussions reported in literature~\cite{vandennest2010,
  vandennest2011}, I choose the weak simulation: as shown in
Fig.~\ref{fig:scheme}, the algorithm will be a probabilistic algorithm
generating samples; such samples will have the same probability
distribution of the measurement of the quantum state. A similar
approach has been proposed for circuit complexity
classes~\cite{viola2010}.

The second problem is that, in many cases, the shortest algorithm that
samples the required probability distribution is just the algorithm
that \emph{simulates} the quantum circuit $C$ acting on
$\left|0\right>$, by brute force. Indeed, such an algorithm is quite
short, although it could take an extremely long time. However, this
solution does not express at all the intuitive concept of
``information content of the output quantum state''. We can get closer
to the intuitive idea by imposing a limit to the calculation time. As
shown in Fig.~\ref{fig:scheme}, I consider an algorithm running in
polynomial time in the number $n$ of qubits and taking an input $I$
representing the quantum state in a compressed (zipped) form. This
input $I$ will then represent the classical information contained in
the quantum state.

Summarizing, we aim at classically simulating the measurements of the
quantum state $C\left|0\right>$, i.e. the application of a circuit $C$
to the initial state $\left|0\right>$ (upper part of
Fig.~\ref{fig:scheme}). We devise a procedure to describe
$C\left|0\right>$ in a compressed form $I$. This representation $I$ is
fed as input to the probabilistic sampling algorithm, which generates
samples with the same probability distribution of the quantum
computer. The aim is to study (lower bound of) the size of $I$,
after imposing that the sampling algorithm works in polynomial time in
the number of qubits $n$. The formalization of this procedure is
described in Sect.~\ref{sect:definition:sampling:algorithm}.

In general, any state $\left|\Psi\right>$ can be obtained by the
application of a suitable circuit $C$ to $\left|0\right>$; the size of
$I$ can be thus very large, exponential in the number of
qubits. However, many of the surprising features of quantum computers
arise in relatively short circuits: calling $m$ the number of 1- and
2-qubit gates of the circuit, I assume that $m$ is polynomial in the
number of qubits $n$. Under this assumption, it makes sense to wonder
if $I$ really needs to be exponentially large in $n$, or if a
polynomially-sized $I$ suffices. The discussion of this question, for
various circuits, is the main topic of this paper.

Literature results on boson sampling and instantaneous quantum
polynomial protocol~\cite{aaronson2010, bremner2010, aaronson2016}
show that it is actually possible to devise specific circuits that
generate a quantum state with a complexity that is
exponentially-growing in the number of qubits $n$.

On the other hand, in Sect.~\ref{sect:examples} I show that some
well-known quantum circuits $C$ lead to quantum states with a
polynomial amount of information, i.e. $I$ has polynomial size. Some
of the cases are not trivial and quite surprising: it could be
na{\"i}vely expected that an exponentially-sized $I$ is needed. Such
quantum circuits are used by quantum algorithms showing an exponential
speed-up with respect to classical computation.

I finally suggest that the final answer to the question ``is this huge
Hilbert space really there?'' is positive, although the possibility of
actually using this huge amount of information is not necessarily
exploited by quantum computers, not even in algorithms showing
superior performances with respect to classical computers.  This fact
can be compared with the above-mentioned fact that generating a large
superposition of states, with a large entanglement, and make them
interfere are necessary ingredients~\cite{vidal2003} for getting
superior performances from quantum computers, although they are not
sufficient~\cite{jozsa2003} for making a quantum circuit hard to be
simulated.

\section{Definition of classical information contained in the quantum
  state generated by a quantum circuit}

\label{sect:definition:sampling:algorithm}

The problems that I will discuss concern the possibility of simulating
the result of a measurement of the quantum state generated by a
quantum circuit $C$. The ``simulation'' is done by a classical
probabilistic algorithm (also called ``randomized''), i.e. a classical
algorithm that has a generator of random uniformly distributed
bits. This concept is captured by the following definition.

\begin{definition}[Sampling algorithm]
  \label{def:sampling:algorithm}
  A set of quantum circuits $\mathcal{C}$ is given, with the property
  that every quantum circuit $C\in \mathcal{C}$, operating on $n$
  qubits, has $m$ 1- and 2-qubit gates, with $m$ polynomial in $n$.  A
  sampling algorithm $A_{\mathcal{C}}$ is a classical probabilistic
  algorithm with the following property. For there exists an input $I$
  to $A_{\mathcal{C}}$, such that the output of $A_{\mathcal{C}}$ is a
  string $x$ having the same probability distribution of measuring $x$
  after applying $C$ to the quantum state $\left|0\right>$.
\end{definition}
Since $A_{\mathcal{C}}$ is probabilistis, its output is different every time it is
run, hence we can speak of the probability of the output.  It is worth
remarking that the term ``classical'' is omitted from the name
``sampling algorithm'' defined above and will be implicitly assumed in
the following.

There is a trivial sampling algorithm working for every $\mathcal{C}$:
giving a representation of $C$ as input $I$, the sampling algorithm is
any algorithm that performs a weak simulation of the quantum
circuit. In this case, the size of input $I$ is polynomial in $n$, but
the calculation time is likely exponential in $n$. Indeed, it is
widely believed that {\bf BQP} contains problems outside {\bf P}: this
means that the weak simulation of some quantum circuits requires
super-polynomial time.  Actually, all the known simulation algorithms
operate in exponential time, although the base of the exponential can
be quite small in specialized algorithms~\cite{bravy2016}.

On the other hand, it is possible to devise a sampling algorithm that
works in polynomial time in $n$, by providing a suitable
representation of $C\left|0\right>$ as input $I$. The procedure is
based on the so-called ``inverse transform sampling'' and is
described, for the present application, in Sect.~III of
Ref.~\cite{hillmich2020}.

The procedure is now briefly sketched. All the probabilities $P(x)$
are first calculated based on the strong simulation of $C$ on the
initial state $\left|0\right>$. Then the probabilities are summed to
get the cumulative probabilities $P_C(x)= \sum_i^x P(i)$ and their
representation is passed to the sampling algorithm $A_{\mathcal{C}}$ as input
$I$. Finally, $A_{\mathcal{C}}$ generates a random number $p$; applies the
bisection algorithm to find the largest $x$ such that $P_C(x)\le p$;
and returns this $x$.

Although this method works for every circuit class $\mathcal{C}$, the
size of input $I$ is exponential in $n$. It is also worth noting that,
in this case, the complexity of the calculation performed by the
quantum circuit $C$ shows up during the calculation of $I$ from $C$,
which requires exponential time.

These two simple examples of sampling algorithms show that there is a
trade-off between space and time resources: if one between input size
and calculation time is polynomial in $n$, then the other is
super-polynomial. For this reason, I suggest that it is interesting to
define the information content of the quantum state generated by a set
of quantum circuits as follows.  .

\begin{definition}
  \label{def:amount:information}
  A set of quantum circuits $\mathcal{C}$ is given, with the
  properties described in
  Def.~\ref{def:sampling:algorithm}. Considering the sampling
  algorithms $A_{\mathcal{C}}$ that operates in polynomial time in $n$
  on average, we say that the information content of the output
  quantum state $C\left|0\right>$ is polynomial in $n$ if there is a
  sampling algorithms $A_{\mathcal{C}}$ for which the input $I$ has
  polynomial size in $n$. Analogously, we define ``super-polynomial''
  and ``exponential''. 
\end{definition}

Literature results on boson sampling and instantaneous quantum
polynomial protocol~\cite{aaronson2010, bremner2010, aaronson2016}
show that it is actually possible to devise specific circuit sets
$\mathcal{C}$ that generate a quantum state with a complexity that is
exponentially-growing in the number of qubits $n$. An explicit
construction of such quantum circuits can be found in Sect. 4.3 of
Ref.~\cite{aaronson2010}; the proof is based on the conjecture that
the polynomial hierarchy does not collapse (information based on a
personal communication by the Author of the cited paper).

On the other hand, in the next section I will present some cases in
which the information content is polynomial. These cases include
quantum circuits used in quantum algorithms that show exponential
speed-up with respect to classical computation.

\section{Examples}

\label{sect:examples}

\subsection{Quantum circuits that can be efficiently simulated by classical
algorithms}

\label{sect:circuits:can:be:simulated}

It is known that some classes of quantum circuits can be efficiently
simulated by classical algorithms, although it is believed that this
cannot be done in general.  In such cases, the input $I$ to the
sampling algorithm $A_{\mathcal{C}}$ can be - trivially - a representation of the
circuit.

A very simple example is the class of circuits formed by a sequence of
Hadamard gates, followed by a sequence of Toffoli gates; these
circuits are called HT in Ref.~\cite{vandennest2010}.
Each Hadamard gate puts a qubit in an uniform
superposition of states $\left|0\right>$ and $\left|1\right>$: this
operation can be simulated by tossing a coin, obtaining a random
$x$. The subsequent Toffoli gates are, actually, classical reversible
gates: applying them to $x$, the desired sample is obtained. The HT
circuits actually correspond to the probabilistic classical
computation, in which a sequence of classical operations is applied
to an input state which also contains random coins. 

Various classes of circuits can be efficiently simulated by classical
algorithms~\cite{vandennest2011}; they include circuits composed by
Clifford gates~\cite{gottesman1998} and by nearest-neighbor
matchgates~\cite{valiant2002}, optical quantum information
circuits~\cite{bartlett2002}, 2-local commuting
circuits~\cite{ni2013}, circuits with small
entanglement~\cite{vidal2003}, tree-like circuits~\cite{markov2008}.
In several cases, the weak simulation of the circuit class has been
explicitly addressed, see the discussion in
Refs.~\cite{vandennest2010, vandennest2011}.

In this context, a particularly important class is represented by the
circuits composed by the Clifford gates: the Gottesmann-Knill
theorem~\cite{gottesman1998} states that they can be efficiently
simulated by classical algorithms. The proof is based on the
``normalizer'' formalism; it will be discussed in more in detail in
Sect.~\ref{sect:clifford:generalization}, together with its
extensions. A different proof, given in Ref.~\cite{vandennest2010},
shows more explicitly that it is possible to make a weak simulation,
i.e. to sample the outcome of the measurements.

In the generic formulation exposed at the beginning of the section,
the input $I$ is a representation of the circuit itself. In the case
of circuits composed by Clifford gates, there is an alternative
possibility.  Indeed, following the various proofs of Gottesmann-Knill
theorem, it appears that the quantum state (after the application of
any number of Clifford gates) can be represented by means of a data
structure with polynomial size in $n$: either the so-called
``tableau''~\cite{aaronson2004}, or a ``graph-state''
representation~\cite{anders2006}. The sampling algorithm can then
operate in polynomial time on such structures for generating the
random samples. The above-mentioned alternative possibility is thus to
use one of such structures as input $I$ for the sampling algorithm.

It must be remarked that the ``tableau'' and the ``graph-state''
representation allows the sampling algorithm to generate samples
independently of the number of gates $m$ of the circuit, which are
thus not limited to be polynomial in $n$: this is a stronger property
than required by Def.~\ref{def:amount:information}, which only refers to $m$
polynomial in $n$. This property is shared with other circuits that
can be classically simulated; the discussion is reported in
Ref.~\cite{vandennest2011}, where the the quantum states that can be
efficiently sampled by classical circuits (independently of the number
of gates $m$) are called "computationally tractable states".

\subsection{Period-finding circuit in Shor's algorithm.}
\label{sect:shor:alg}

It could be argued that quantum circuits that cannot be simulated
classically should exploit quantum states with a super-polynomial
information content.  The factorization of a large number $N$ is
believed to be impossible for a classical algorithm working in
polynomial time in $\log N$.  The famous Shor's factorization
algorithm uses a quantum circuit that operates in polynomial number of
gates. We thus expect that the quantum states involved in this quantum
circuit cannot be sampled by a classical algorithm. I will show that
this is actually not the case, i.e., also in the case of Shor's
algorithm, the information content is polynomial in $n$.

A fundamental element of Shor's factorization algorithm is the
``modulus exponentiation'' function $f$:
\begin{equation}
  f_{a,N}(x) = a^x \pmod N
\end{equation}
defined based on $N$ and an integer $a$. This function can be
efficiently calculated by a classical reversible circuit (thus, also
by a quantum one).

The quantum circuit used in Shor's algorithm has the role of finding
the period $r$ of $f_{a,N}(x)$, for given $a$ and $N$. The function is
indeed periodic: there are two numbers $r$ and $x_{min}$ such that
$f_{a,N}(x) = f_{a,N}(x+r)$ for every $x\ge x_{min}$. I define
``period'' as the smallest $r$.

In Shor's algorithm, $a$ and $N$ are coprime: this condition ensures
that $f_{a,N}(x)$ is always periodic, i.e. $x_{min}=0$. In the
following, I will also focus on this case. In more general cases,
$f_{a,N}(x)$ is still ``ultimately periodic'', i.e., there are values
of $a$ and $N$ such that $f_{a,N}(x)$ is periodic for $x$ larger
than a $x_{min}>0$. There are no reasons why such values cannot be
used anyway in the period-finding quantum circuit. This case can be
addressed by a trivial extension of the methods discussed in this
section and is reported in detail in Appendix~\ref{app:non:periodic}.

I describe now the quantum circuit. It uses two registers, of $n_x$
and $n_f$ qubits, respectively, initialized to
$\left|0\right>\left|0\right>$.

\begin{circuit}
  \label{circ:shor}
  \begin{enumerate}
  \item Apply Hadamard gates to the second register to get a uniform
    superposition of $x$:
    \begin{equation}
      \left|\Psi_1\right> = 2^{-n_x}
      \sum_{x} \left|x\right>\left|0\right>
    \end{equation}
  \item Calculate the modulus exponentiation:
    \begin{equation}
      \left|\Psi_2\right> = 2^{-n_x}
      \sum_{x} \left|x\right>\left|f_{a,N}(x)\right>
    \end{equation}
  \item Apply the quantum Fourier transform to the first register:
    \begin{equation}
      \left|\Psi_3\right> = 2^{-n_x}
      \sum_{x} \left|\tilde{x}\right>\left|f_{a,N}(x)\right>
    \end{equation}
  \end{enumerate}
\end{circuit}
  
The circuit is hard-wired to perform the modulus exponentiation $f$
with given $a$ and $N$. It is worth noticing that the quantum Fourier
transform used here is the one on $\mathbb{Z}_{2^{n_x}}$, and that it
operates on an interval of numbers whose length is $2^{n_x}$.

When the measurement is finally performed, the second register results
in a value of $f_0$, equal to $f_{a,N}(\bar{x})$ for a random
$\bar{x}$. In $\left|\Psi_2\right>$, the first register thus contains
a superposition of all the $\left|x\right>$ such that $f_{a,N}(x)=f_0$
(all the $x$ in the pre-image of $f_0$, including $\bar{x}$): they are
exactly spaced by $r$. Applying the Fourier transform to the first
register then transforms this periodicity into peaks, at multiples of
$2^{n_x}/r$. The final measurement of the first register will give an
$\tilde{x}$ close to one of such peaks: the periodicity $r$ of $f$ is
found.

Before discussing the tricky details, it is worth considering the
special case when $r$ is a power of 2 (smaller than $2^{n_x}$). In
this case, the peaks in the Fourier transform are exactly at integer
numbers $2^{n_x}/r$ and are perfectly sharp: the
Fourier transform is 0 everywhere, except for values of $\tilde{x}$
multiples of $2^{n_x}/r$, where it takes complex values with same modulus.
Measuring the the first register thus always gives exactly a
$\tilde{x}$ multiple of $r$.

In this special case, it is thus easy to weakly simulate the output of
the measurement of the quantum circuit Circ.~\ref{circ:shor}, by the
following procedure:
\begin{description}
\item[First register] output a random multiple of $2^{n_x}/r$ less than $2^{n_x}$;
\item[Second register] generate a random $\bar{x}$ and output $f_{a,N}(\bar{x})$
\end{description}

Unfortunately, the situation is less simple when $r$ is not a
power of 2: the position of the peaks in the Fourier transform,
$2^{n_x}/r$, do not exactly fall on an integer number and the peaks
become broad. The measurement of the first register will give
an integer value close to a multiple of $2^{n_x}/r$.

As discussed above, assuming that the second register will be
measured, after the calculation of the modulus exponentiation, the
first register contains a superposition of states $\left|r m +
x_0\right>$, where $0\le m\le M$ and $M$ is the maximum $m$ such that
$r m + x_0<2^{n_x}$; the number of the elements in the superposition
is $M+1$. The Fourier transform is the superposition of all the
$\Psi(\tilde{x})\left|\tilde{x}\right>$ with:
\begin{equation}
  \Psi\left(\tilde{x}\right) = \frac{1}{\sqrt{2^{n_x}\left(M+1\right)}}
  \sum_{m=0}^M \exp\left(2\pi i \frac{\tilde{x}\left(r m+x_0\right)}
      {2^{n_x}} \right)
\end{equation}
I remind that the frequency $\tilde{x}$ is in the same range of $x$,
i.e. $0\dots 2^{n_x}-1$.  The expression, up to an irrelevant phase
term depending on $x_0$, is a geometrical series, which can be
explicitly calculated:
\begin{equation}
  \Psi\left(\tilde{x}\right) = \frac
   { \exp\left(2\pi i \frac{\tilde{x}x_0} {2^{n_x}} \right)}
  {\sqrt{2^{n_x}\left(M+1\right)}}
  \frac{
  \exp\left[2\pi i \frac{\tilde{x} r \left(M+1\right)}{2^{n_x}} \right]-1
  }{
  \exp\left[2\pi i \frac{\tilde{x} r}{2^{n_x}} \right]-1
  }
  \label{eq:shor:fourier:raw}
\end{equation}
This expression is not defined for $\tilde{x}=0$; here and in the
following, I will implicitly assume that the undefined values are
extrapolated by continuity, considering a real $\tilde{x}\to 0$.

The probability
$P\left(\tilde{x}\right)=\left|\Psi\left(\tilde{x}\right)\right|^2$
is now calculated; it is useful to write it as:
\begin{equation}
  P\left(\tilde{x}\right) = \frac{q}{2^{n_x}} \rho\left(\tilde{x} p\right)
\end{equation}
where:
\begin{equation}
  \rho\left(v\right) =
  \frac{1}{q\left(M+1\right)}
  \frac{
    1-\cos\left[2\pi \frac{v \left(M+1\right)}{q} \right]
  }{
    1-\cos\left[2\pi \frac{v}{q} \right]
  }
\end{equation}
and $p/q=r/2^{n_x}$; the two natural numbers $p$ and $q$ are chosen so that
they are coprime.

Thanks to the periodicity of the trigonometric functions, $\rho$ is
periodic with period $q$. This suggests to ``wrap'' it as:
\begin{equation}
  P\left(\tilde{x}\right) =
  \frac{q}{2^{n_x}} \rho\left(\tilde{x} p \bmod q\right)
\end{equation}
It is possible to express $\tilde{x}$, in the range $0\dots 2^{n_x}-1$, as
\begin{equation}
  \tilde{x} =s\left(\tilde{x}\right)+z\left(\tilde{x}\right) q
\end{equation}
where:
\begin{equation}
  s\left(\tilde{x}\right) = \tilde{x} \bmod q
\end{equation}
is in the range $0\dots q-1$ and
\begin{equation}
  z\left(\tilde{x}\right) = \left\lfloor \frac{\tilde{x}}{q} \right\rfloor
\end{equation}
is in the range $0\dots 2^{n_x}/q-1$. The result is:
\begin{equation}
  P\left(\tilde{x}\right) = \frac{q}{2^{n_x}}
  \rho\left[s\left(\tilde{x}\right) p \bmod q\right]
\end{equation}
This expression can be interpreted by identifying the two terms on the
right-hand side as two probabilities:
and
\begin{equation}
  P\left(\tilde{x}\right) = \rho_z\left[z\left(\tilde{x}\right)\right]
    \rho_s\left[s\left(\tilde{x}\right)\right]
\end{equation}
where the probability distributions of $s$ and $z$ are:
\begin{equation}
  \rho_z\left(z\right) = \frac{q}{2^{n_x}}
\end{equation}
and
\begin{equation}
  \rho_s\left(s\right) =
  \rho\left[s p \bmod q\right]
\end{equation}
Since $p$ and $q$ are coprime, $v=s p \bmod q$ is a bijection between
$v$ and $s$ on the range $0\dots q-1$; moreover it is possible to
efficiently calculate $s$ given $v$. Sampling $s$ from
$\rho_s\left(s\right)$ can thus be accomplished by first sampling $v$
from $\rho\left(v\right)$ and then calculating $s$ using the
above-mentioned bijection.

Summarizing, sampling from $P\left(\tilde{x}\right)$, i.e. simulating
Circ.~\ref{circ:shor}, can be accomplished by the following algorithm:
\begin{algorithm} \,
  \label{alg:shor}
  
  \begin{enumerate}
  \item Sample a natural number $v$ from the distribution
    $\rho\left(v\right)$, with $v$ in the range $0\dots q-1$;
  \item Calculate $s$ such that $v=s p \bmod q$;
  \item Uniformly sample a natural number $z$ in the range $0\dots 2^{n_x}/q-1$;
  \item Return $\tilde{x}=s+zq$.
  \end{enumerate}
\end{algorithm}

\begin{figure}
  \includegraphics{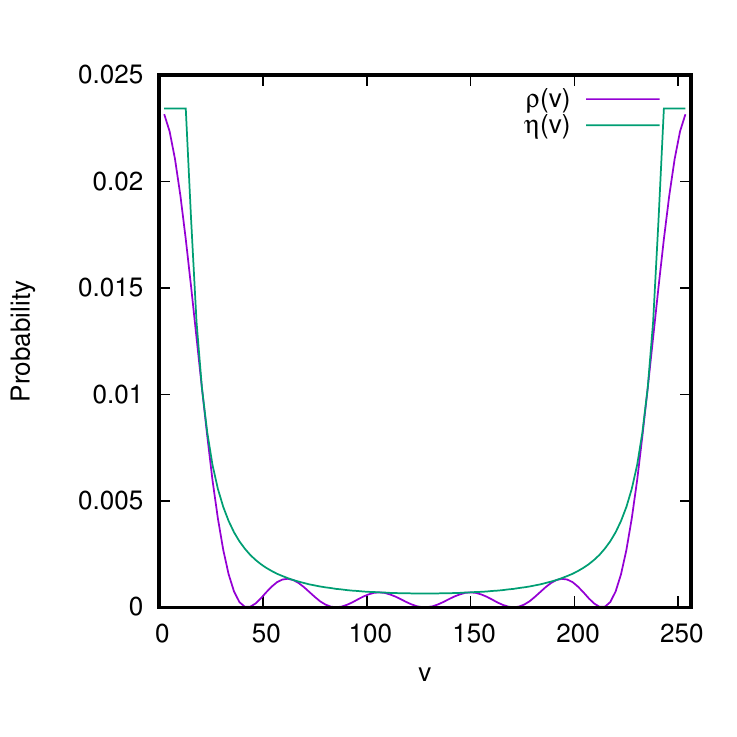}
  \caption{Graph of $\rho\left(v\right)$ and $\eta\left(v\right)$, for
    $M=5$ and $q=2^8$. }
  \label{fig:rho:rho:e}
\end{figure}

It remains to show how to sample from
$\rho\left(v\right)$. Figure~\ref{fig:rho:rho:e} shows an example of
the graph of $\rho\left(v\right)$: intuitively, it represents the
probability between two consecutive peaks.  It is possible to
efficiently sample from $\rho\left(v\right)$ by means of Monte Carlo
rejection sampling, using the proposal function $\eta\left(v\right)$:
\begin{equation}
  \eta\left(v\right) = \frac{1}{q\left(M+1\right)}
  \min \left[ \left(M+1\right)^2 ,
    \frac{ 2 }{
      1-\cos\left(2\pi \frac{v}{q} \right)
    }
    \right]
  \label{eq:def:eta}
\end{equation}
The proposal function $\eta\left(v\right)$ is approximately the
envelope of $\rho\left(v\right)$; both functions are shown in
Fig.~\ref{fig:rho:rho:e}. The application of the Monte Carlo procedure
by means of the reported proposal function $\eta\left(v\right)$ does
not present remarkable features and is sketched in
Appendix~\ref{app:sampling:rho} for completeness sake, together with
the (quite long but trivial) proof of its efficiency.

\subsection{Circuit used in Grover's algorithm}

The Grover's algorithm is presented as a quantum algorithm for finding
a database entry $x_0$ with a desired property among the $n$-bit
strings~\cite{grover1997, lavor2003}.
The requested property is expressed by a function $f(x) : [0,1]^n \to
[-1,1]$: it returns -1 if $x$ has the requested property, else 1.  The
goal of the algorithm is to find $x_0\in [0,1]^n$ such that
$f(x_0)=-1$. I assume that there is only a single $x_0$ with this
property.

The function $f(x)$ is seen as a black box, i.e. nothing is known
about it, except that it is implemented as a subroutine that returns
$f(x)$ on input $x$. A classical algorithm should make an exhaustive
search on $2^n$ elements to find $x_0$: on average, it will take
$2^{n-1}$ attempts. Grover's algorithm shows a polynomial seed-up,
namely the number of steps is of the order of $2^{n/2}$.

In this section, I show that the quantum states arising in the quantum
circuit used in Grover's algorithm carry a polynomial classical
information amount: actually, I give two different possible sampling
algorithms $A_{\mathcal{C}}$ and input states $I$, exploiting
different resources. Rigorously, Grover's algorithm prescribes the use
of a quantum circuit with length exponential in the number of bits
$n$, which would not fit into Def.~\ref{def:sampling:algorithm}. Thus
I must considered a shortened version of that circuit. A complete
description of Grover's algorithm is outside the scope of this paper;
I describe the shortened quantum circuit only to the extent that is
needed to discuss the information content according to
Def.~\ref{def:amount:information}.

The shortened version of the quantum circuit is the following.

\begin{circuit}[Shortened version of the circuit used in Grover's algorithm]
  \label{circ:grover}
  The function $f(x)$ is used to define the operator $U_f$:
  \begin{equation}
    U_f \left|x\right> = f(x) \left|x\right>
  \end{equation}
  A second operator, called ``diffusion operator'', $U_g$, is defined
  analogously, with a function $g(x)$, which takes the value -1 for
  $x=0$, else 1.

  The initial quantum state is $\left|0\right>$. Hadamard gates are
  applied to every qubit. Then the following sets of gates are
  repeatedly applied $t$ times:
  \begin{itemize}
  \item $U_f$ is applied
  \item Hadamard gates are applied to all the qubits
  \item $U_g$ is applied
  \item Hadamard gates are applied to all the qubits
  \end{itemize}

  The number $t$ is a polynomial in $n$, $t(n)$.

  The output of the circuit is finally measured in the computational
  basis.
\end{circuit}

Definition~\ref{def:sampling:algorithm} requires that the number of 1-
and 2-qubit gates of the circuit is polynomial in $n$. The diffusion
operator $U_g$ can be built using a polynomial (in $n$) number of 1-
and 2-qubit gates~\cite{lavor2003}. Moreover, I assume that $f(x)$ can
be calculated by a classical reversible circuit with a polynomial (in
$n$) number of 1- and 2-qubit gates; also $U_f$ can be realized with a
polynomial number of gates. In Grover's algorithm, $t$ is exponential
in $n$; here I assume instead that it is polynomial, so that the total
number of 1- and 2-qubit gates $m$ is polynomial in $n$, as requested
in Def.~\ref{def:sampling:algorithm}.

It can be shown that, at any $m$, the quantum state is a superposition
of $\left|0\right>$ and $\left|x_0\right>$
states~\cite{grover1997}. The probability $P(x)$ of getting the
required result $x_0$ can be explicitly calculated:
\begin{equation}
  P(t) = \sin^2 \left[
    \left(2t+1\right) \arcsin\left(\frac{1}{2^{n/2}}\right) \right]
  \label{eq:grover:def:P:t}
\end{equation}
It is worth remarking that $t$ is considered as a fixed polynomial in
$n$.

A possible sampling algorithm $A_{\mathcal{C}}$ can be the following.
\begin{algorithm} \,
  \label{alg:grover:with:x0}
  
  {\bf Input} $x_0$ such that $f(x_0)=-1$ and the number of steps $t$.
  
  {\bf Output} $x_0$ with probability $P(t)$, else 0.
  
  (The algorhm is trivial)
\end{algorithm}

The described procedure requires to first find $x_0$, thus the
algorithm that generates $I$ must perform an exhaustive search on
$[0,1]^n$. A more lazy alternative consists to pass a representation
of $f(x)$ (or $U_f$) as input $I$.

\begin{algorithm} \,
  \label{alg:grover:without:x0}
  
  {\bf Input} A representation of the circuit that calculates $f(x)$
  and the number of steps $t$.
  
  {\bf Output} $x_0$ with probability $P(t)$, else 0.
  
  \begin{enumerate}
  \item Randomly generate $N$ numbers $x_i$ in the range $0\dots 2^n$, where:
    \begin{equation}
      N = \left\lceil \frac{\log\left[1-P(t)\right]}
      {\log\left(1-\frac{1}{2^n}\right)} \right\rceil
      \label{eq:alg:grover:2:N}
    \end{equation}
  \item Calculate $f(x_i)$;
  \item If $f(x_i)$ is 1 for every $x_i$, return 0.
  \item Call $x_0$ the $x_i$ such that $f(x_i)=-1$.
  \item With probability $P'$ return $x_0$, where
    \begin{equation}
      P' = \frac{P\left(t\right)}{
        1-\left(1-\frac{1}{2^n}\right)^N
      }
      \label{eq:alg:grover:2:Pp}
    \end{equation}
  \item Else return 0.
  \end{enumerate}
\end{algorithm}

This algorithm is correct and efficient. The proof is sketched in
Appendix~\ref{app:grover:proof}. The appendix also shows that $N$
(proportional to the running time) is polynomial in $t$, with leading
order two.

Both algorithms, Alg.~\ref{alg:grover:without:x0} and
Alg.~\ref{alg:grover:with:x0}, can be used as sampling algorithm
$A_{\mathcal{C}}$ in Def.~\ref{def:amount:information}. The case of
Alg.~\ref{alg:grover:without:x0} is similar to the case discussed in
Sect.~\ref{sect:circuits:can:be:simulated}: it shows that
Circ.~\ref{circ:grover} can be efficiently simulated by classical
algorithms, although it takes quadratic time in $n$ and $m$.  The
alternative sampling algoritym, Alg.~\ref{alg:grover:with:x0},
requires that the solution of the database search, $x_0$, is given as
input $I$: it reminds the situation of Sect.~\ref{sect:shor:alg},
where the solution of the problem (there, the period $r$) is used as
input to the sampling algorithm $A_{\mathcal{C}}$.

The two algorithms also differ in the used resources: there is a
trade-off between them. Algorithm~\ref{alg:grover:without:x0} requires
a possibly longer $I$, representing the circuit of $U_f$, runs in
quadratic time in $t$, and the preparation of $I$ takes polynomial
time in the size of the circuit of $U_f$. By contrast,
Alg.~\ref{alg:grover:with:x0} requires a shorter $I$, just containing
$x_0$ and $t$, runs in fixed time independently on $t$, but the
preparation of $I$ requires an exhaustive search to find $x_0$.

\subsection{Generalization of Clifford gates}

\label{sect:clifford:generalization}

The Clifford gates, mentioned above in
Sect.~\ref{sect:circuits:can:be:simulated}, can be generalized by a
method based on group theory~\cite{vandennest2013}.  A complete and
detailed discussion of the generalization of Clifford gates is outside
the scope of this paper. Here I discuss the topic only to the extent
that it is relevant for the information content of
Def.~\ref{def:amount:information}. The discussion follows the lines of
Ref.~\cite{vandennest2013}.

A finite Abelian group $G$ is given. The quantum states are defined on
a computational basis formed by states $\left|x\right>$, with $x\in
G$. The ``generalized Clifford gates'' $C_i$ are a set of unitary
operators, defined based on $G$.  These gates include the Fourier
transforms over $G$. Roughly, they operate on many bits, so they look
more like whole circuits. The traditional Clifford gates are generated
by the group $G=\mathbb{Z}_2^n$. The definition of these gates is not
reported here and is not necessary for the discussion.

Based on $G$, the following definition is given.
\begin{definition}[Coset state of $G$]
  An Abelian finite group $G$ is given. $K$ is a subgroup of $G$ and
  $x$ is an element of $G$.  The coset state $\left|\psi(K,x)\right>$
  is:
  \begin{equation}
    \left|\psi\right> = \frac{1}{\sqrt{\left|K\right|}}
    \sum_{k\in K} \left|k+x\right>
  \end{equation}
  where $\left|\cdot\right|$ is the cardinality of a set.
\end{definition}

The interest in the generalized Clifford gates and in the coset states
stems from the fact that it is possible to efficiently classically
sample from the output of any sequence of generalized Clifford gates
$C_i$ applied to a coset state $\left|\psi(K,x)\right>$.  However this
operation (the efficient classical sampling) requires to know a
suitable representation of $\left|\psi(K,x)\right>$: in particular, a
generating set of $K$ must be known, with polynomial size in
$\log\left|G\right|$ (roughly analogous to the number of qubits,
although here we do not have qubits at all). The proof is given in
Ref.~\cite{vandennest2013}.

In order to apply this result to the discussion performed in this paper,
I introduce a class of circuits.
\begin{circuit}[Circuits ending with generalized Clifford gates]
  \label{circ:generalized:clifford}
  An Abelian finite group $G$ is given. The circuit is composed by two
  stages. The first, $U$, is applied to $\left|0\right>$ and generates
  a coset state $\left|\psi(K,x)\right>$. The second is composed by a
  sequence of generalized Clifford gates over $G$.
\end{circuit}
According to the above-mentioned results of
Ref.~\cite{vandennest2013}, the following efficient sampling algorithm
can be used.
\begin{algorithm} \,
  \label{alg:generalized:clifford}
  
  {\bf Input} The generating set of $K$, with polynomial size in
  $\left|G\right|$.

  {\bf Output} The requested samples

  The algorithm is described in Ref.~\cite{vandennest2013}.
\end{algorithm}

We see that the quantum circuits of the family
Circ.~\ref{circ:generalized:clifford} can be efficiently sampled by
Alg.~\ref{alg:generalized:clifford}, thus also in this case the amount
of information is polynomial in $n$ according to
Def.~\ref{def:amount:information}.

Quantum circuits of the family Circ.~\ref{circ:generalized:clifford}
are used in quantum algorithms for solving the ``Abelian hidden
subgroup problems''~\cite{lomont2004}. A problem in this class (rigorously,
with some limitations) can be expressed in terms of a decision problem
in {\bf NP} that is in {\bf BQP} but is believed to lie outside {\bf
  P}. The certificate is the generating set of the hidden
subgroup. Thus we see that the certificate itself, which is
polynomially large in $\log\left|G\right|$, is the required input $I$.
This situation reminds what happens with the factorization discussed
in Sect.~\ref{sect:shor:alg}, where the period $r$ is used to
calculate the factors, constituting the certificate of the {\bf NP}
problem, and is given as input $I$ to the sampling algorithm
$A_{\mathcal{C}}$. Actually, the factorization can be seen as an Abelian hidden
subgroup problems by means of a quantum circuit of the family
Circ.~\ref{circ:generalized:clifford}. Shor's algorithm is indeed an
approximated version of such a quantum algorithm.

\subsection{Period-finding circuit with uniform superposition of $N$}

In the previous section, the value of $N$ is hard-wired inside the
quantum circuit and the information needed by the sampling algorithm,
in order to describe the quantum circuit, is the period $r$. It is
however possible to make an equivalent quantum circuits that takes $N$
as an input, in a separate register. This single change would not
significantly modify the situation: depending on the input $N$, the
sampling algorithm is provided by a different $r$. Apparently, the
situation would dramatically change if the input $N$ is constructed as
a superposition of all the possible values of qubits, by means of
Hadamard gates: providing the value of $r$ for each $N$ would require
an amount of bits exponentially large in $\log N$. This would be a
case of polynomial size of information according to
Def.~\ref{def:amount:information}. I will show that this is not the case and
the difficulty can be circumvented, leading to a polynomial amount of
information also in this case.

We use three registers, initially containing 0, representing $N$, $x$,
and the result of the modulus exponentiation $f_{a,N}(x)$; the number
of qubits are $n_N$, $n_x$ and $n_f$, respectively.  The stages of the
circuit are:

\begin{circuit}
  \label{circ:shor:multiple}
  \begin{enumerate}
  \item Apply Hadamard gates to the first register to get a uniform
    superposition of $N$:
    \begin{equation}
      \left|\Psi_1\right> = 2^{-n_N}
      \sum_N \left|N\right>\left|0\right>\left|0\right>
    \end{equation}
  \item Apply Hadamard gates to the second register to get a uniform
    superposition of $x$:
    \begin{equation}
      \left|\Psi_2\right> = 2^{-n_Nn_x}
      \sum_{N,x} \left|N\right>\left|x\right>\left|0\right>
    \end{equation}
  \item Calculate the modulus exponentiation:
    \begin{equation}
      \left|\Psi_3\right> = 2^{-n_Nn_x}
      \sum_{N,x} \left|N\right>\left|x\right>\left|f_{a,N}(x)\right>
    \end{equation}
  \item Apply the quantum Fourier transform to the second register:
    \begin{equation}
      \left|\Psi_4\right> = 2^{-n_Nn_x}
      \sum_{N,x} \left|N\right>\left|\tilde{x}\right>\left|f_{a,N}(x)\right>
    \end{equation}
  \end{enumerate}
\end{circuit}

Sampling the measurement of $\left|\Psi_4\right>$ would give a random
$N$ in the first register, a value $f_{a,N}(x)$ for a random $x$ in
the third, and an $\tilde{x}$ close to a multiple of $r$. Sampling
first a random $N$ and then trying to sample $\tilde{x}$ would require to
calculate $r$ from $N$, which is likely not feasible in polynomial time,
or to provide the sampling algorithm with the value of $r$ for every $N$,
which is not feasible in exponential space in $\log N$.

However, it turns out that it is possible to sample the couple
$(N,r)$, with a uniform distribution of $N$, such that $r$ is the
period of $f_{a,N}(x)$, in polynomial time in $\log N$. Although
surprising, it is actually easier to sample the couple $(N,r)$ rather
than finding the period $r$ corresponding a number $N$.  I give a
proof of this curious fact in the following of the section.  This fact
allows the sampling algorithm of Sect.~\ref{sect:shor:alg} to be
adapted to this new case, thus the quantum states involved in the
above described quantum circuit, Circ.~\ref{circ:shor:multiple}, have
a polynomial amount of information according to
Def.~\ref{def:amount:information}.

Curiously, with the described procedure, no information must be
passed to the sampling algorithm, except the number of qubits $n_N$,
$n_x$ and $n_f$: the outcome of the measurement can be sampled by a
fixed algorith, with a fixed length, whatever is the number of
involved qubits. By seeking for a more complex quantum state, we
actually devised a much simpler one.

In the rest of this section, I sketch the algorithm to sample the
couple $(N,r)$ in polynomial time. This algorithm profits on a method
for uniformly sampling numbers in factorized form~\cite{bach1988,
  kalai2003}, operating in polynomial time: although uniformly
sampling a number $N$ and then factorizing it is believed to be a
difficult task, directly generating the factors $p_i$, such that their
product $N$ is uniformly distributed, can be accomplished in
polynomial time in the number of desired bits.

The algorithm given by Kalai~\cite{kalai2003} is extremely simple.

\begin{algorithm}[Kalai's algorithm] \,
  \label{alg:kalai}
  
  {bf Input:} maximum number $N_{max}$
  
  {\bf Output:} a random $N$, uniformly distributed between
  1 and $N_{max}$, along with its prime factorization.
  
  \begin{enumerate}
  \item Generate a sequence $N_{max} \ge s_1 \ge s_2 \ge \dots s_l =
    1$ by choosing $1\le s_1 \le N_{max}$ and $1\le s_{i+1}\le s_i$,
    until reaching 1.
  \item Select the $s_i$ that are prime.
  \item Let $N$ be the product of the prime $s_i$'s.
  \item If $N\le N_{max}$, output $N$ (along with its factors) with
    probability $N/N_{max}$.
  \item Otherwise, repeat from 1.
  \end{enumerate}
\end{algorithm}

First, we notice that the Kalai's algorithm can be
adapted to generate an uniformly distributed $N$ along with the
factorization of $N-1$, with a slight modification. Moreover, we
notice that the prime numbers generated by the algorithm are the
results of a uniform sampling of numbers $s_i$.

The algorithm for sampling the couples $(N,r)$ first makes use of the
Kalai's algorithm (Alg.~\ref{alg:kalai}) for uniformly sampling $N$
(along with its prime factors) in the range $1\dots
N_{max}=2^{n_N}$. In turn, the uniform sampling of the $s_i$'s,
required in step 1 of Kalai's algorithm, is also performed by using
the Kalai's algorithm, so that the obtained prime
factors $p_i$ of $N$ are supplied with the factorization of
$p_i-1$. So, summarizing, we uniformly sample $N$, in the form:
\begin{equation}
  N = \prod_i^{P} p_i^{\nu_i}
\end{equation}
where the $p_i$'s are prime; moreover, we know the prime factorization
of $p_i-1$:
\begin{equation}
  p_i = \prod_j^{Q_i} q_{i,j}^{\mu_{i,j}} + 1
  \label{eq:factorization:p:i}
\end{equation}
where the $q_i$'s are prime.  It is worth highlighting that Kalai's
algorithm (Alg.~\ref{alg:kalai}) is used at two levels, not
recursively, i.e. we need the prime factorization of $N$ and of
$p_i-1$, but \emph{we do not need} the prime factorization of
$q_{i,j}-1$.  The information obtained by the procedure, the $p_i$'s
and the $q_{i,j}$'s, is sufficient to calculate $r$ as follows.

In group theory, the quantity $r$ is known as the \emph{order of $a$} in the
multiplicative group $\mathbb{Z}^*_N$.  The \emph{order} of the group,
$\varphi(N)$, is defined as the number of invertible elements of
$\mathbb{Z}^*_N$. It can be computed from the prime factorization of
$N$:
\begin{equation}
  \varphi(N) = \prod_i^{P} p_i^{\nu_i-1} \left(p_i - 1\right)
\end{equation}
By using Eq.~\ref{eq:factorization:p:i}:
\begin{equation}
  \varphi(N) = \prod_i^{P} p_i^{\nu_i-1} \prod_j^{Q_i} q_{i,j}^{\mu_{i,j}}
\end{equation}
Some of the $p_i$'s and $q_i$'s can be equal. In polynomial time, it
is possible to rearrange the two products as:
\begin{equation}
  \varphi(N) = \prod_i^{Q'} {q'}_i^{\lambda_i}
\end{equation}
where the $q'_i$'s are distinct primes and $Q'$ is the number of
distinct $q'_i$'s.


We then proceed by finding the largest $\tau_i\le \lambda_i$ such that:
\begin{equation}
  a^{\frac{\varphi(N)}{{q'}_i^{\tau_i}}} \equiv 1 \pmod{N}
\end{equation}
This operation can be also done in polynomial time in $n_N$. The order
of $a$ is then:
\begin{equation}
  r = \prod_i^{Q'} q_i^{\lambda_i-\tau_i}
\end{equation}
This works since $\mathbb{Z}^*_N$ is an Abelian group which decomposes
as a direct group of Abelian groups of orders
${q'}_1^{\lambda_1}, \ldots, q_{Q'}^{\lambda_{Q'}}$.

Summarizing, the reported procedure first generates uniformly
distributed random numbers $N$ along with the prime factors $p_i$ and
$q_{i,j}$; then it calculates $r$ from $p_i$ and $q_{i,j}$. We thus
conclude that the quantum states involved in the above-described
circuit, Circ.~\ref{circ:shor:multiple}, also contain a polynomial
amount of information according to Def.~\ref{def:amount:information}.

\section{Conclusion}

\label{sect:conclusion}

I proposed how to define the information content of a quantum state,
inspired to the Kolmogorov complexity. The formal definition is given
in Def.~\ref{def:amount:information}. It is trivial to see that the
information content of any quantum state is at most exponential in the
number of qubits $n$.  I discuss the information content of quantum
states generated by applying $m$ quantum gates to the initial state
$\left|0\right>$, with $m$ polynomial in $n$, as often done in quantum
computation.

Trivially, the information content is polynomial in the case of
quantum circuits that can be simulated classically. More interesting
is the case of quantum circuits that are believed to be difficult to
classically simulate, in particular the ones used in algorithms that
are believed to show superior performances in quantum computers than
in classical ones. In particular, I analyzed quantum circuits that are
used to solve {\bf NP} problems and thus cannot be simulated
classically. They include the algorithms for solving the Abelian
hidden subgroup problem and the factorization (Shor's
algorithm). Grover's algorithm is also discussed. It turns out that
the information content roughly corresponds to the certificate of the
corresponding {\bf NP} problem, which has polynomial size in $n$.

Looking for a more complex situation, it can be suggested to try with
problems in which such a certificate does not exist. I created such a
case by using the quantum circuit used to factorize $N$ and feeding it
with a superposition of every $N$ with $n$ bits. In this case, the
information content should contain the factorizations of every $N$,
thus an exponential amount of information.  Instead, even in this case
the amount of information is polynomial in $n$: actually, it turns out
that the sampling algorithm requires even less information than for
the case of the factorization with a fixed $N$.

We know that there are specially designed quantum circuits giving rise
to quantum states with a super-polynomial information amount,
according to Def.~\ref{def:amount:information}. However, from the
examples discussed above, we see that the use of quantum states with
this huge amount of information is not necessarily needed by quantum
algorithms showing an exponential speed-up with respect to classical
ones.  It seems that the huge Hilbert space is really there, but the
speed-up of quantum computation is not necessarily require to store a
correspondingly huge amount of information in the quantum state.

\newpage

\renewcommand{\thepage}{Appendix~--~\arabic{page}}

\appendix

\section{Modulus exponentiation with generic $a$ and $N$ }

\label{app:non:periodic}

In this section, I discuss the period-finding quantum circuit used in
Shor's algorithm, for the case in which the state is prepared using
the modulus exponentiation $f_{a,N}(x)$ with generic $a$ and $N$. In
other words, I relax the condition (present in Shor's algorithm) that
$a$ and $N$ are coprime.

I remind that, in general, the modulus exponentiation is
``ultimately'' periodic: there are two numbers $r$ and $x_{min}$ such
that $f_{a,N}(x)=f_{a,N}(x+r)$ for every $x\ge x_{min}$. If $a$ and
$N$ are coprime (as in Shor's algorithm), then $x_{min}=0$. In this
section I consider the more general case, in which $x_{min}$ can be
larger than 0, i.e. the function is ``ultimately periodic'' but not
necessarily ``periodic''.

In order to sample the result of the measurement, it is first
necessary to calculate $x_{min}$; following the discussion in
Sect.~\ref{sect:shor:alg}, we assume that $r$ is known. It is of
course possible to find $x_{min}$ by simply trying all the values
starting from 0 and verifying if
$f_{a,N}(x_{min})=f_{a,N}(x_{min}+r)$. This procedure is actually
efficient, thanks to the following proposition:

\begin{proposition}
  For every $N$ and $a<N$, $x_{min}$ is smaller than the number of
  binary digits of $N$, $x_{min}\le \log_2(N)$.
\end{proposition}
\begin{proof}
  Let us define:
  \begin{equation}
    d_k = \gcd\left(a^k, N\right)
  \end{equation}
  With increasing $k$, $d_k$ increases; however, it cannot exceed $N$,
  thus it will stabilize. This means that there is a $\bar{k}$ such
  that $d_k=d_{\bar{k}}$ for every $k\ge\bar{k}$. I call
  $d=d_{\bar{k}}$ this $\gcd$.

  In order to proceed with the discussion, it is worth expressing
  $a$ and $N$ in terms of their prime factorization:
  \begin{equation}
    a = \prod_{i} p_i^{\mu_i} q_i^{\nu_i}
  \end{equation}
  and
  \begin{equation}
    N = \prod_{i} p_i^{\lambda_i} s_i^{\tau_i}
  \end{equation}
  where $p_i$, $q_i$, and $s_i$ are distinct primes and $\mu_i$, $\nu_i$,
  $\lambda_i$, and $\tau_i$ are the multiplicities. The primes $p_i$ are
  the only common prime factors between $a$ and $N$.

  I highlight that the prime factorization of $a$ and $N$ is only used
  in this proof, in order to prove the thesis, but it is not proposed
  as a part of an algorithm. Indeed, such an algorithm would not be
  efficient.

  Using the prime factorization, $d_k$ can be expressed as:
  \begin{equation}
    d_k = \prod_{i} p_i^{\min(k\mu_i, \lambda_i)}
    \label{eq:app:d:k:factorized}
  \end{equation}
  For large enough $k$ we get $d$:
  \begin{equation}
    d = \prod_{i} p_i^{\lambda_i}
  \end{equation}
  The ratio between $N$ and $d$ is thus:
  \begin{equation}
    \frac{N}{d} = \prod_{i} s_i^{\tau_i}
  \end{equation}
  It can be noticed that $N/d$ does not have any prime factor in
  common with $a$, thus they are coprime. According to Euler's
  theorem, the function $a^x \bmod N/d$ is periodic (not just
  ultimately periodic), thus there is an integer $r'>0$ such that:
  \begin{equation}
    a^{r'} \bmod \frac{N}{d} = a^0 \bmod \frac{N}{d} = 1
  \end{equation}
  In other words, there is a natural number $p>0$ such that:
  \begin{equation}
    a^{r'} = \frac{Np}{d} + 1
  \end{equation}
  Moreover, according to the definition of $d$, there is a natural
  number $q>0$ such that:
  \begin{equation}
    a^{\bar{k}} = d q
  \end{equation}
  Combining the last two equations:
  \begin{equation}
    a^{\bar{k}+r'} = N p + a^{\bar{k}}
  \end{equation}
  Applying the modulus:
  \begin{equation}
    a^{\bar{k}+r'} \bmod N = a^{\bar{k}} \bmod N
  \end{equation}
  By comparing this expression with the definitions of $r$ and
  $x_{min}$, we conclude that $r'$ is a multiple of $r$ and
  \begin{equation}
    x_{min}\le \bar{k} .
    \label{eq:app:x:min:bar:k}
  \end{equation}

  From Eq.~\ref{eq:app:d:k:factorized}, if $k$ is equal or larger than
  all the $\lambda_i$, then $d_k=d$. As a consequence, $\bar{k}$ is at
  most the multiplicity of any prime in the factorization of $N$. In
  turn, this multiplicity is at most $\log_2(N)$. This can be easily
  seen by noticing that the multiplicity of 2 in any number up to $N$
  is up to $\lfloor\log_2(N)\rfloor$; it is not possible to have a
  larger multiplicity with prime factors larger than 2. The conclusion
  is thus that:
  \begin{equation}
    \bar{k}\le \log_2(N)
  \end{equation}
  The combination of this last equation and
  Eq.~\ref{eq:app:x:min:bar:k} leads to the thesis.
\end{proof}

With the knowledge of $x_{min}$ and $r$, it is now possible to
describe the sampling procedure. Following the procedure described in
Sect.~\ref{sect:shor:alg}, we can still generate a random $\bar{x}$
and output $f_{a,N}(\bar{x})$ as the measurement of the second
register, even if $x_{min}$ does not vanish. However, two different
cases must be taken into consideration to properly sample the
measurement of the first register.
\begin{description}
\item[$\bar{x}<x_{min}$] The pre-image of $f_{a,N}(\bar{x})$ under
  $f_{a,N}(x)$ is constituted by the single point $\bar{x}$ ($\bar{x}$
  belongs to the non-periodic part of $f_{a,N}(x)$). The Fourier
  transform of the state $\left|\bar{x}\right>$ is the uniform
  superposition of all the $\tilde{x}$, thus the requested sample is a
  uniformly random $\tilde{x}$.
\item[$\bar{x}\ge x_{min}$] The pre-image of $f_{a,N}(\bar{x})$ under
  $f_{a,N}(x)$ is composed by a sequence of $M+1$ points, spaced by
  $r$, represented as $m r+x_0$ for $m$ in the range $0\dots M$. The
  values of $M$ and $x_0$ can be easily calculated based on $x_{min}$
  and $r$. The Fourier transform corresponds to
  Eq.~\ref{eq:shor:fourier:raw}; the procedure reportd in
  Sect.~\ref{sect:shor:alg} can thus be applied, with the suitable
  values of $M$.
\end{description}

\section{Sampling from the probability $\rho\left(v\right)$}

\label{app:sampling:rho}

In this section I show how to sample from $\rho\left(v\right)$.  This
operation can be efficiently performed by means of Monte Carlo
rejection sampling.  For clarity sake and for introducing the
notation, I shortly summarize the Monte Carlo rejection sampling
procedure. A ``proposal function'' $\eta\left(v\right)$, such that
$\eta\left(v\right)\ge \rho\left(v\right)$, is suitably chosen; then
the following steps are performed:
\begin{enumerate}
\item A natural number $v$ (the proposed sample) is sampled from the
  proposal probability
  $P_p\left(v\right)=\frac{\eta\left(v\right)}{\sum_{v'} \eta\left(v'\right)}$;
\item The number $v$ is accepted (i.e., returned as the valid
  sample) with the acceptance probability
  $ P_a\left(v\right)=\frac{\rho\left(v\right)}{\eta\left(v\right)}$;
\item In case of rejection, the procedure is repeated from point 1.
\end{enumerate}

The chosen $\eta\left(v\right)$ is reported in Eq.~\ref{eq:def:eta};
it is approximately the envelope of $\rho\left(v\right)$. The two
functions are shown in Fig.~\ref{fig:rho:rho:e}.

It is easy to prove that $\eta\left(v\right)\ge \rho\left(v\right)$;
this inequality can be also seen in the graph. Under this condition,
the Monte Carlo rejection sampling gives the correct results. However,
it is still necessary to prove that the procedure is efficient,
i.e. if it operates in polynomial time in $\log(q)$ on average. This
is ensured if~\cite{robert}:
\begin{itemize}
\item It is possible to efficiently sample from the proposal
  probability $P_p\left(v\right)$.
\item The number
  \begin{equation}
    \nu=\sum_{v} \eta\left(v\right) ,
    \label{eq:def:nu}
  \end{equation}
  representing the average number of proposed samples needed to obtain
  one accepted sample, is at most polynomial in $\log(q)$;
\end{itemize}

The efficient method for sampling the proposal probability
$P_p\left(v\right)$ that I propose is based on the algorithm for
sorted (i.e. monotonically decreasing) probability
distributions~\cite{bringmann2017}. Our proposal probability
$P_p\left(v\right)$ is not sorted, but it is symmetrical around the
central point $v=q/2$, i.e.  $P_P\left(v\right)=P_P\left(q-v\right)$, and
it is monotonically decreasing on the left side and
increasing on the right side (see Fig.~\ref{fig:rho:rho:e}). It is
thus possible to sort the probabilities in decreasing order, by
sequentially taking one point from the left and one point from the
right. This is done by defining:
\begin{equation}
  v\left(\bar{v}\right) = \left\{ \begin{array}{lr}
    \left\lfloor \frac{\bar{v}}{2} \right\rfloor & \bar{v}\, \mathrm{even} \\
    q-1-\left\lfloor \frac{\bar{v}}{2} \right\rfloor & \bar{v}\, \mathrm{odd}
  \end{array} \right.
\end{equation}
The resulting $P_p\left[v\left(\bar{v}\right)\right]$ monotonically
decreases with increasing $\bar{v}$; moreover, the function
$v\left(\bar{v}\right)$ is a bijection on the domain $0\dots q-1$. It
is thus possible to apply the the efficient algorithm of
Ref.~\cite{bringmann2017} to the sorted probability distribution
$P_p\left[v\left(\bar{v}\right)\right]$; the algorithm returns
$\bar{v}$, which is then used to calculate $v\left(\bar{v}\right)$ and
return it as the desired sample.

The proof of the second point is now sketched. We start by identifying
the point $\delta$ in which the transition between the two arguments
of the $\min(\cdot, \cdot)$ operator in Eq.~\ref{eq:def:eta} takes
place:
\begin{equation}
  \delta = 
  \frac{q}{2\pi}
  \arccos\left[1-\frac{2}{\left(M+1\right)^2} \right]
  \label{eq:def:delta}
\end{equation}
This allows us to rewrite Eq.~\ref{eq:def:eta} as:
\begin{equation}
  \eta\left(v\right) = \left\{
  \begin{array}{lr}
    \frac{M+1}{q}
    &
    v\in R_A\cup R_E\\ 
    \frac{1}{q\left(M+1\right)}
    \frac{ 2 }{ 1-\cos\left(2\pi \frac{v}{q} \right) } &
    v\in R_B\cup R_C\cup R_D
  \end{array}
  \right.
  \label{eq:def:eta:ranges}
\end{equation}
where the five ranges are $R_A = 0\dots \lfloor\delta\rfloor$,
$R_B = \lfloor\delta\rfloor+1\dots q/2-1$, $R_C = \{q/2\}$ (a single number),
$R_D = q/2+1\dots q-\lfloor\delta\rfloor-1$, and
$R_E = q-\lfloor\delta\rfloor\dots q-1$. The reason for splitting the central
part of the domain into the three ranges $R_B$, $R_C$, and $R_D$ is
related to the symmetry of $\eta\left(v\right)$ and will be discussed
below.

According to the definition, $M$ is in the range $1\dots q-1$; using
Eq.~\ref{eq:def:delta}, we get $\delta\in \left(0,q/6\right]$. These
bounds will be used in the following calculations.

The calculation of $\nu$ by Eq.~\ref{eq:def:nu:ranges} is now split
into the ranges:
\begin{equation}
  \nu=S_1 + S_2 + S_3
  \label{eq:def:nu:ranges}
\end{equation}
where $S_1$, $S_2$, and $S_3$ are the sumations of
$\eta\left(v\right)$ over $R_A\cup R_E$, $R_B\cup R_D$, and $R_C$,
respectively.

The addends $\eta\left(v\right)$ are constant over $R_A$ and $R_E$, and
equal to $(M+1)/q$:
\begin{equation}
  S_1 = \sum_{v\in R_A\cup R_E} \eta\left(v\right) =
  \frac{M+1}{q}\left(2 \left\lfloor\delta\right\rfloor +1\right)
  \label{eq:def:S:one:def}
\end{equation}
The range $R_C$  is composed by a single number:
\begin{equation}
  S_3 = \sum_{v\in R_C} \eta\left(v\right) =
  \frac{1}{q\left(M+1\right)}
  \label{eq:def:S:three:def}
\end{equation}
Thanks to the symmetry $\eta\left(v\right)=\eta\left(q-v\right)$, the
sumations of $\eta\left(v\right)$ over $R_B$ and $R_D$ are equal. The
result is:
\begin{equation}
 S_2 = \sum_{v\in R_B\cup R_D} \eta\left(v\right) =
 \frac{4}{q\left(M+1\right)} \sum_{v=\lfloor\delta\rfloor+1}^{q/2-1} \frac{1}{1-\cos
   \left(2\pi \frac{v}{q}\right)}
\end{equation}
The sum is further split into two parts, separating the first term,
which gives the largest contribution:
\begin{equation}
  S_2 = S_{2a} + S_{2b}
  \label{eq:def:nu:subranges}
\end{equation}
where:
\begin{equation}
  S_{2a} = \frac{4}{q\left(M+1\right)} \frac{1}{1-\cos
    \left(2\pi \frac{\left\lfloor\delta\right\rfloor+1}{q}\right)}
  \label{eq:S:two:a:definition}
\end{equation}
and
\begin{equation}
  S_{2b} = \frac{4}{q\left(M+1\right)} \sum_{v=\lfloor\delta\rfloor+2}^{q/2-1}
  \frac{1}{1-\cos \left(2\pi \frac{v}{q}\right)}
  \label{eq:S:two:b:definition}
\end{equation}

It is now necessary to find quantities, larger than $S_1$, $S_{2a}$,
$S_{2b}$, and $S_3$, that are independent of $M$.

Since $M$ and $q$ are positive and integer, from
Eq.~\ref{eq:def:S:three:def}:
\begin{equation}
  S_3 \le 1
  \label{eq:S:three:result}
\end{equation}

$S_1$ is defined in Eq.~\ref{eq:def:S:one:def}. Since $M\le q-1$:
\begin{equation}
  S_1 \le 2 \frac{M+1}{q} \left\lfloor\delta\right\rfloor + 1
  \label{eq:S:one:first:passage}
\end{equation}
Applying the inequality:
\begin{equation}
  \arccos\left(1-\right) \le \sqrt{2x}+x
\end{equation}
to Eq.~\ref{eq:def:delta}, we get:
\begin{equation}
  \delta \le
  \frac{q}{\pi}\left( \frac{1}{M+1} +
  \frac{1}{\left(M+1\right)^2} \right)
  \label{eq:delta:inequality}
\end{equation}
Using this inequality to further elaborate
Eq.~\ref{eq:S:one:first:passage} gives:
\begin{equation}
  S_1 \le 1 + \frac{3}{\pi}
  \label{eq:S:one:result}
\end{equation}

Now I discuss the term $S_{2a}$, defined in
Eq.~\ref{eq:S:two:a:definition}. The inequality
\begin{equation}
  1-\cos \left(x\right)
  \ge \frac{x^2}{4}
\end{equation}
holds for $0\le x \le 3\pi/4$. From
Eq.~\ref{eq:delta:inequality}, $\delta\le 3q/(4\pi)$; the argument of
$\cos(\cdot)$ in Eq.~\ref{eq:S:two:a:definition},
i.e. $2\pi\left(\lfloor\delta\rfloor+1\right)/q$, is within the range
for large enough $q$, thus the inequality can be used:
\begin{equation}
  S_{2a} \le \frac{4}{\pi^2} \frac{q}{M+1}
  \frac{1}{\left(\left\lfloor\delta\right\rfloor+1\right)^2}
  \label{eq:S:two:a:passage:a}
\end{equation}
Applying the inequality:
\begin{equation}
  \arccos\left(1-x\right) \ge \sqrt{x}
\end{equation}
to Eq.~\ref{eq:def:delta}, we get:
\begin{equation}
  \delta \ge
  \frac{q}{2\pi}
  \frac{\sqrt{2}}{M+1}
  \label{eq:delta:inequality:bis}
\end{equation}
Moreover, $\lfloor\delta\rfloor+1\ge\delta$. These inequalities are
used to elaborate Eq.~\ref{eq:S:two:a:passage:a}:
\begin{equation}
  S_{2a} \le 8\frac{M+1}{q}
\end{equation}
Remembering that $M\le q-1$:
\begin{equation}
  S_{2a} \le 8 .
  \label{eq:S:two:a:result}
\end{equation}

The last term to be discussed is $S_{2b}$, defined in
Eq.~\ref{eq:S:two:b:definition}.  The addends of the summation
decrease with increasing $v$; this allows us to calculate a quantity
greater than $S_{2b}$ in terms of an integral:
\begin{equation}
  S_{2b} \le 
  \frac{4}{q\left(M+1\right)} \int_{\lfloor\delta\rfloor+1}^{q/2-2}
  \frac{\mathrm{d}\xi}
  {1-\cos \left(2\pi \frac{\xi}{q}\right)}
\end{equation}
The calculation of the integral is straightforward:
\begin{equation}
  S_{2b} \le
  \frac{2}{\pi\left(M+1\right)}
  \left[
    \cot\left( \pi \frac{\left\lfloor\delta\right\rfloor+1}{q}  \right)
    - \tan\left( \frac{2\pi}{q} \right)
    \right]
\end{equation}
For $q>2$ the second term in parenthesis is positive:
\begin{equation}
  S_{2b} \le
  \frac{2}{\pi\left(M+1\right)}
  \cot\left( \pi \frac{\left\lfloor\delta\right\rfloor+1}{q}  \right)
\end{equation}
The argument of the $\cot(\cdot)$ is between
0 and $\pi$ (actually, less than $\pi/2)$; in this range:
\begin{equation}
  \cot\left(x\right) \le \frac{1}{x}
\end{equation}
Using this inequality:
\begin{equation}
  S_{2b} \le
  \frac{2}{\pi^2\left(M+1\right)}
  \frac{q}{\left\lfloor\delta\right\rfloor+1}
\end{equation}
Using Eq.~\ref{eq:delta:inequality:bis}:
\begin{equation}
  S_{2b} \le
  \frac{2\sqrt{2}}{\pi}
  \label{eq:S:two:b:result}
\end{equation}

Summarizing, the terms in which $\nu$ is decomposed, i.e. $S_1$,
$S_{2a}$, $S_{2b}$, and $S_3$ (see Eqs.~\ref{eq:def:nu:ranges} and
\ref{eq:def:nu:subranges}) are smaller than constant quantities, as
shown in Eqs.~\ref{eq:S:three:result}, \ref{eq:S:one:result},
\ref{eq:S:two:a:result}, and \ref{eq:S:two:b:result}. The sampling of
$\rho\left(v\right)$ is thus efficient.

\section{Proof of the validity and efficiency of
  Algorithm~\ref{alg:grover:without:x0}}
\label{app:grover:proof}

In this section I prove that Alg.~\ref{alg:grover:without:x0} is
valid and efficient.

The probability of \emph{not} finding $x_0$ among the $N$ generated
samples $x_i$ is $\left(1-2^n\right)^N$, corresponding to the
probability of returning 0 in step 3. The probability of reaching step
4 is
\begin{equation}
  P_4=1-\left(1-2^n\right)^N
\end{equation}
It represents the probability of finding $x_0$. According to
Eq.~\ref{eq:alg:grover:2:N}, $N$ is the smallest integer such that
$P_4\ge P(t)$: the algorithm finds $x_0$ with a slightly larger
probability than the quantum circuit (but with a polynomially longer
time, see below). This reasoning also ensures that $P'$ in
Eq.~\ref{eq:alg:grover:2:Pp} is smaller than 1.

The total probability of returning $x_0$ in step 5 is $P_4 P'$,
i.e. the composed probability of reaching step 5 and outputting $x_0$.
Using $P'$ defined in Eq.~\ref{eq:alg:grover:2:Pp}:
\begin{equation}
  P_4 P' = P(t)
\end{equation}
Thus I have shown that the algorithm outputs $x_0$ with the desired
probability. Necessarily 0 is also output with the correct probability
$1-P(t)$.

It remains to show that the sampling is efficient. This depends on how
many attempts must be done for finding $x_0$, i.e. on $N$.  By using
the definition of $N$, Eq.~\ref{eq:alg:grover:2:N}, and the definition
of $P(t)$, Eq.~\ref{eq:grover:def:P:t}:
\begin{equation}
  N \le  \frac{\log\left\{1-
    \sin^2 \left[
      \left(2t+1\right) \arcsin\left(\frac{1}{2^{n/2}}\right) \right]
    \right\}}
  {\log\left(1-\frac{1}{2^n}\right)} +1
\end{equation}
For large enough $n$, $\arcsin\left(1/2^{n/2}\right)<1/2^{n/2-1}$:
\begin{equation}
  N \le  \frac{\log\left[1-
    \sin^2 \left( \frac{2t+1}{2^{n/2-1}} \right)
    \right]}
  {\log\left(1-\frac{1}{2^n}\right)} +1
\end{equation}
For positive $x$, $\sin(x)<x$:
\begin{equation}
  N \le  \frac{\log\left[1-
    \left( \frac{2t+1}{2^{n/2-1}} \right)^2
    \right]}
  {\log\left(1-\frac{1}{2^n}\right)} +1
\end{equation}
Since $-\log(1-x)>x$:
\begin{equation}
  N \le -2^n \log\left[1-
    \left( \frac{2t+1}{2^{n/2-1}} \right)^2
    \right] +1
\end{equation}
For any polynomial $t(n)$ and for large enough $n$:
\begin{equation}
  N \le 8 \left(2t+1\right)^2 + 1
\end{equation}
Since $t(n)$, the number of iterations, is a polynomial in $n$, and
each iteration is polynomial in $n$, the running time is polynomial in
$n$, i.e. the algorithm is efficient.







\tableofcontents

\bibliographystyle{unsrt}
\bibliography{quantum}

\end{document}